\documentclass[11pt]{article}
\usepackage{hyperref}
\usepackage{epsfig}
\usepackage{fullpage}
\usepackage{latexsym}
\usepackage{amssymb,amsfonts,amsmath,amsthm}

\def\01{\{0,1\}}

\newcommand{\eps}{\varepsilon}
\renewcommand{\epsilon}{\eps}

\newcommand{\ket}[1]{|#1\rangle}

\newcommand{\ketbra}[2]{|#1\rangle\langle#2|}

\newcommand{\inpc}[2]{\langle{#1},{#2}\rangle} 
\newcommand{\Tr}{\mbox{\rm Tr}}

\newcommand{\R}{\mathbb{R}}

\newcommand{\dist}{{\rm dist}}
\newcommand{\supp}{{\rm supp}}
\newcommand{\calA}{{\cal A}}
\newcommand{\calP}{{\cal P}}

\newtheorem{theorem}{Theorem}
\newtheorem{lemma}[theorem]{Lemma}

\newtheorem{observation}[theorem]{Observation}

\renewcommand{\qed}{\hfill{\rule{2mm}{2mm}}}
\renewenvironment{proof}[1][]{\begin{trivlist}
\item[\hspace{\labelsep}{\bf\noindent Proof#1:\/}] }{\qed\end{trivlist}}

\begin{document}

\title{Upper Bounds on the Noise Threshold for Fault-tolerant Quantum Computing}
\author{Julia Kempe\thanks{School of Computer Science, Tel-Aviv University, Tel-Aviv 69978, Israel. Supported by
the European Commission under the Integrated Project Qubit Applications (QAP) funded by the IST directorate
as Contract Number 015848, by an Alon Fellowship of the Israeli Higher Council of Academic Research and by an
Individual Research Grant of the Israeli Science Foundation.}
\and
Oded Regev\thanks{School of Computer Science, Tel-Aviv University, Tel-Aviv 69978, Israel. Supported  by the Binational Science Foundation, by the Israel Science Foundation, and  by the European Commission under the Integrated Project QAP funded by the IST directorate as Contract Number 015848.}
 \and Falk Unger\thanks{Partially supported by the European Commission under the Integrated Project Qubit Applications (QAP) funded by the IST directorate as Contract Number 015848.} \and Ronald de Wolf\thanks{Partially supported by a Veni grant from the Netherlands Organization for Scientific Research (NWO), and by the European Commission under the Integrated Project Qubit Applications (QAP) funded by the IST directorate as Contract Number 015848.}}
\maketitle

\begin{abstract}
We prove new upper bounds on the tolerable level of noise in a quantum circuit.
We consider circuits consisting of unitary $k$-qubit gates each of whose input wires is
subject to depolarizing noise of strength $p$, as well as arbitrary one-qubit gates that
are essentially noise-free. We assume that the output of the circuit is the result of
measuring some designated qubit in the final state.
Our main result is that for $p>1-\Theta(1/\sqrt{k})$, the output of any such circuit of large enough
depth is essentially independent of its input, thereby making the circuit useless.
For the important special case of $k=2$, our bound is $p>35.7\%$.
Moreover, if the only allowed gate on more than one qubit is the two-qubit CNOT gate, then our bound becomes $29.3\%$.
These bounds on $p$ are notably better than previous bounds, yet are incomparable because
of the somewhat different circuit model that we are using.
Our main technique is the use of a Pauli basis decomposition, which we believe
should lead to further progress in deriving such bounds.
\end{abstract}

\section{Introduction}\label{secintro}

The field of quantum computing faces two main tasks:
to build a large-scale quantum computer, and to figure out what it can do once it exists.
In general the first task is best left to (experimental) physicists and engineers,
but there is one crucial aspect where theorists play an important role, and that is in
analyzing the level of noise that a quantum computer can tolerate before breaking down.

The physical systems in which qubits may be implemented are typically tiny and fragile
(electrons, photons and the like). This raises the following paradox: On the one hand we
want to isolate these systems from their environment as much as possible, in order to avoid the
noise caused by unwanted interaction with the environment---so-called ``decoherence''. But on the other hand
we need to manipulate these qubits very precisely in order to carry out computational operations.
A certain level of noise and errors from the environment is therefore unavoidable in any implementation,
and in order to be able to compute one would have to use techniques of error correction and fault tolerance.

Unfortunately, the techniques that are used in classical error correction and fault tolerance
do not work directly in the quantum case. Moreover, extending these techniques
to the quantum world seems at first sight to be nearly impossible due to the continuum of possible
quantum states and error patterns.
Indeed, when the first important quantum algorithms were discovered~\cite{bernstein&vazirani:qcomplexity,simon:power,shor:factoring,grover:search}, many
dismissed the whole model of quantum computing as a pipe dream, because it was
expected that decoherence would quickly destroy the necessary quantum properties of
superposition and entanglement.

It thus came as a great surprise when, in the mid-1990s, \emph{quantum error correcting codes} were developed
by Shor and Steane~\cite{shor:scheme,steane:errcor}, and these ideas later led to the development of schemes
for \emph{fault-tolerant quantum computing}~\cite{shor:faulttol,klz:treshold,KnillLZ98,
aharonov&benor:faulttol,kitaev:qcsurvey,gottesman:thesis}. Such schemes take any quantum algorithm designed
for an ideal noiseless quantum computer, and turn it into an implementation that is robust against noise, as
long as the amount of noise is below a certain threshold, known as the \emph{fault-tolerant threshold}. The
overhead introduced by the fault-tolerant schemes is typically quite modest (a polylogarithmic factor in the
total running time of the algorithm).

The existence of fault-tolerant schemes turns the problem
of building a quantum computer into a hard but possible-in-principle engineering problem:
if we just manage to store our qubits and operate upon them with a level of noise below
the fault-tolerant threshold, then we can perform arbitrarily long quantum computations.
The actual {\em value} of the fault-tolerant threshold is far from determined, but will have
a crucial influence on the future of the area---the more noise a quantum computer
can tolerate in theory, the more likely it is to be realized in practice.\footnote{The ``fault-tolerant
threshold" is actually not a universal constant, but rather depends on the details of the
circuit model (allowed set of gates, type of noise, etc.). A more precise discussion
will be given later.}

The first fault-tolerant schemes were only able to tolerate noise on the order of $10^{-6}$,
which is way below the level of accuracy that experimentalists can hope
to achieve in the foreseeable future.
These initial schemes have been substantially improved in the past decade.
In particular, Knill has recently developed various schemes which, according to numerical calculations, seem to be
able to tolerate more than 1\%\ noise~\cite{knill:realisticallynoisy,knill:thresholdanalysis}.
If we insist on provable constructions, the best known threshold is on the order
of $0.1\%$~\cite{agp:threshold,aliferis:fibonacci,aliferis:thesis,reichardt:thesis}.

Constructions of fault-tolerant schemes provide a \emph{lower bound} on the fault-tolerant threshold.
A very interesting question, which is the topic of the current paper, is whether one
can prove \emph{upper bounds} on the fault-tolerant threshold. Such bounds give
an indication on how far away we are from finding optimal fault-tolerant schemes.
They can also give hints as to how one should go about constructing improved
fault-tolerant schemes. Such upper bounds are statements of the form ``any quantum
computation performed with noise level higher than $p$ is essentially useless",
where ``essentially useless" is usually some strong indication that interesting
quantum computations are impossible in such a model. For instance, Buhrman et al.~\cite{bcllsu:faulttol}
quantify this by giving a classical simulation of such noisy quantum computation,
and Razborov~\cite{razborov:qnoise} shows that if the computation is too long, the
output of the circuit is essentially independent of its input.

The best known upper bounds on the threshold are $50\%$ by Razborov~\cite{razborov:qnoise} and $45.3\%$ by Buhrman et al.~\cite{bcllsu:faulttol}.
(These bounds are incomparable because they work in different models; See the end of this section for more accurate statements.)
As one can see, there are still about two orders of magnitude between our best upper and lower bounds
on the fault-tolerant threshold. This leaves experimentalists in the dark
as to the level of accuracy they should try to achieve in their experiments.
In this paper, we somewhat reduce this gap.  So far, much more work has
been spent on lower bounds than on upper bounds.  Our approach will be the less-trodden
road from above, hoping to bring new techniques to bear on this problem.

\paragraph{Our model.}
In order to state our results, we need to describe our circuit model.
We consider parallel circuits, composed of $n$ \emph{wires} and $T$ \emph{levels} of gates (see Figure \ref{Fig:circuit}).
We sometimes use the term \emph{time} to refer to one of the $T+1$ ``vertical cuts'' between the levels.
For convenience, we assume that the number of qubits $n$ does not change during the computation.
Each level is described by a partition of the qubits, as well as a gate assigned
to each set in the partition. Notice that at each level, all qubits must go through
some gate (possibly the identity).
Notice also that for each gate the number of input qubits is the same as the number of
output qubits.

\begin{figure}[h]
\center{\epsfxsize=4in \epsfbox{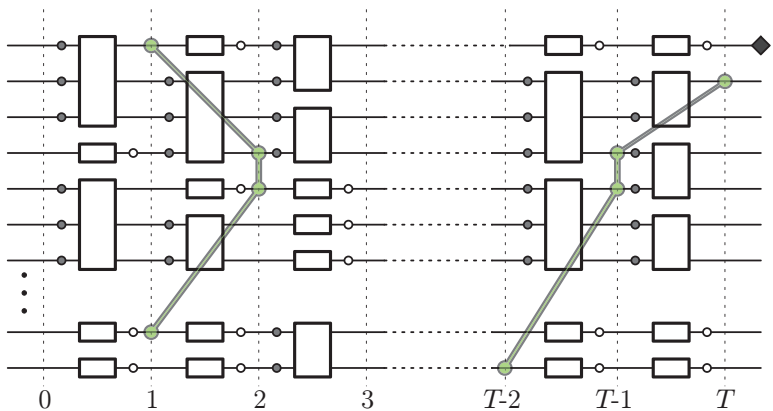}} \caption{Parallel circuit with $k=3$ and $T$ levels.
Dark circles denote $\eps_k$-depolarizing noise,
and light circles denote $\eps_1$-depolarizing noise. Also marked are two consistent sets (defined in Section~\ref{secmain}), each
containing four qubits. The first has distance $1$, the second has distance $T-2$.
The output qubit is in the upper right corner.}
\label{Fig:circuit}
\end{figure}

We assume the circuit is composed of $k$-qubit gates that are probabilistic mixtures of unitary operations,
as well as arbitrary (i.e., all completely-positive trace-preserving) one-qubit gates.
We assume the output of the circuit is the outcome of a measurement of a designated output qubit in the computational basis.
Finally, we assume that the circuit is subject to noise as follows.
Recall that $p$-depolarizing noise on a certain qubit replaces that qubit by the completely mixed
state with probability $p$, and does not alter the qubit otherwise.
Formally, this is described by the superoperator $\cal E$ acting on a qubit $\rho$ as
$
{\cal E}(\rho)=(1-p)\rho+p{I}/2.
$
We assume that each one-qubit gate is followed by at least $\eps_1$-depolarizing noise on its output qubit,
where $\eps_1>0$ is an arbitrarily small constant. Thus one-qubit gates can be essentially noise-free.
We also assume that each $k$-qubit gate is preceded by at least $\eps_k$-depolarizing noise on
each of its input qubits, where $\eps_k > 1-\sqrt{2^{1/k}-1} = 1 - \Theta(1/\sqrt{k})$.

\paragraph{Our results.}

In Section~\ref{secmain} we prove our main result:

\begin{theorem}\label{the:main}
Fix any $T$-level quantum circuit as above.
Then for any two states $\rho$ and $\tau$, the probabilities of obtaining measurement outcome $1$ at the output qubit
starting from $\rho$ and starting from $\tau$, respectively, differ by at most $2^{-\Omega(T)}$.
\end{theorem}

\noindent
In other words, for any $\eta>0$, the probability of measuring $1$ at the output qubit of a circuit running for
$T=O(\log(1/\eta))$ levels is (up to $\pm\eta$) independent of the input.
This makes the output essentially independent of the starting state, and renders long computations
``essentially useless''.

Of special interest from an experimental point of view is the case $k=2$, for which our bound becomes about
$35.7\%$. Furthermore, for the case in which the only allowed two-qubit gate is the CNOT gate, we can improve our
bound further to about $29.3\%$, as we show in Section~\ref{seccnot}. This case is interesting both
theoretically and experimentally. Note also that the CNOT gate together with all one-qubit gates forms
a universal set~\cite{barencoea:gates}.

\paragraph{Significance of results.}
Here we comment on the significance of our results and of our model.

First, it is known that fault-tolerant quantum computation is impossible (for any positive noise level)
without a source of fresh qubits. Our model takes care of this by allowing arbitrary one-qubit gates---in
particular, this includes gates that take any input, and output a fixed one-qubit state, for instance the
classical state $\ket{0}$. This justifies our assumption that the number of qubits in the circuit remains the
same throughout the computation: all qubits can be present from the start, since we can reset them to
whatever we want whenever needed.

Second, our assumption that all $k$-qubit gates are mixtures of unitaries does slightly restrict generality.
Not every completely-positive trace-preserving map can be written as a mixture of unitaries.
However, we believe that it is a reasonable assumption. As one indication of this, to the best of our knowledge, all known
fault-tolerant constructions can be implemented using such gates (in addition to
arbitrary one-qubit gates). Moreover, all known quantum algorithms gain their speed-up over classical
algorithms by using only unitary gates.

A slightly more severe restriction is the assumption
that the output consists of just one qubit. However, we believe that in many instances
this is still a reasonable assumption. For instance, this is the case whenever
the circuit is required to solve a decision problem.
Moreover, our results can be easily extended to deal with
the case in which a small number of qubits are used as an output.

By allowing essentially noise-free one-qubit gates, our model addresses the fact that gates
on more than one qubit are generally much harder to implement.
It should also be noted that the exact value of the constant $\epsilon_1$ is inessential and can be
chosen arbitrarily small, as this just affects the constant in the $\Omega(\cdot)$ of Theorem~\ref{the:main}.
In fact, $\epsilon_1>0$ is only necessary because otherwise it would be possible to let
$\rho:=\ketbra{0}{0}\otimes \rho'$ and $\tau:=\ketbra{1}{1}\otimes
\tau'$, do nothing for $T$ levels (i.e., apply noise-free one-qubit identity gates on all wires)
and then measure the first qubit. The resulting difference between output probabilities is then $1$.
Instead of assuming an $\eps_1>0$ amount of noise, we could alternatively deal with this
issue by requiring that every path from the input to the output qubit goes
through enough $k$-qubit gates. Our proof can be easily adapted to this case.

Note that since our theorem applies to arbitrary starting states, it in particular
applies to the case that the initial state is encoded in some good quantum
error-correcting code, or that it is some sort of ``magic state"~\cite{BK:Magic,R:Magic}.
In all these cases, our theorem shows that the computation becomes essentially
independent of the input after sufficiently many levels.

Finally, it is interesting to note that our bound on the threshold
behaves like $1-\Theta(1/\sqrt{k})$. This matches what is known for classical circuits~\cite{EvansS99,EvansS03},
and therefore probably represents the correct asymptotic behavior. Previous bounds
only achieved an asymptotic behavior of $1-\Theta(1/k)$~\cite{razborov:qnoise}.

\paragraph{Techniques.}
We believe that a main part of our contribution is introducing a new technique
for obtaining upper bounds on the fault-tolerant threshold. Namely,
we use a Pauli basis decomposition in order to track the state of the computation.
We believe this framework will be useful also for further analysis
of quantum fault-tolerance. A finer analysis of the Pauli coefficients might
improve the bounds we achieve here, and possibly obtain bounds that are tailored
to other computational models.

\paragraph{Related work.}
The work most closely related to ours is that of Razborov~\cite{razborov:qnoise}. There, he proves an upper
bound of $\eps_k=1-1/k$ on the fault-tolerant threshold. On one hand, his result is stronger than ours as it
allows arbitrary $k$-qubit gates and not just mixtures of unitaries. Razborov also has a second result,
namely the trace distance between the two states obtained by applying the circuit to starting states $\rho$
and $\tau$, respectively, goes down as $n2^{-\Omega(T)}$ with the number of levels $T$. Hence even the
results of an arbitrary $n$-qubit measurement on the full final state become essentially independent of the
initial state after $T=O(\log n)$ levels. On the other hand, the value of our bound is better for all values
of $k$, and we also allow essentially noise-free one-qubit gates. Hence the two results are incomparable.
Razborov's proof is based on tracking how the trace distance evolves during the computation. Our proof is
similar in flavor, but instead of working with the trace distance, we work with the Frobenius distance (since
it can be easily expressed in terms of the Pauli decomposition).

Buhrman et al.~\cite{bcllsu:faulttol} show that classical circuits can efficiently simulate
any quantum circuit that consists of perfect, noise-free \emph{stabilizer operations}
(meaning Clifford gates (Hadamard, phase gate, CNOT), preparations of states in
the computational basis, and measurements in the computational basis)
and arbitrary one-qubit unitary gates that are followed by $45.3\%$ depolarizing noise.
Hence such circuits are not significantly more powerful than classical circuits.%
\footnote{The $45.3\%$-bound of~\cite{bcllsu:faulttol} is in fact \emph{tight} if one additionally allows perfect
classical control (i.e., the ability to condition future gates on the earlier classical measurement outcomes):
circuits with perfect stabilizer operations and arbitrary one-qubits gates suffering from less than $45.3\%$ noise,
can simulate perfect quantum circuits. See \cite{reichardt:distilling} and \cite[Section~5]{bcllsu:faulttol}.
These assumptions are not very realistic, however.  In particular the assumption that one can implement
perfect, noise-free CNOTs is a far cry from experimental practice.}
This result is incomparable to ours: the noise models
and the set of allowed gates are different (and we feel ours is more realistic).
In particular, in our case noise hits the qubits going into the $k$-qubit gates but barely affects the one-qubit gates,
while in their case the noise only hits the non-Clifford one-qubit unitaries.

Another related result is by Virmani et al.~\cite{vhp:thresholds}. Instead of depolarizing noise, they
consider ``dephasing noise''. This models phase-errors only: while we can view depolarizing noise of strength
$p$ as applying one of four possible operations (I,X,Y,Z), each with probability $p/4$, dephasing noise of
strength $p$ applies one of two possible operations, I or Z, each with probability $p/2$. Virmani et
al.~\cite{vhp:thresholds} show, among other results, that any quantum circuit consisting of perfect
stabilizer operations, and one-qubit unitary gates that are diagonal in the computational basis and are
followed by  dephasing noise of strength  $29.3\%$, can be efficiently simulated classically. Their result is
incomparable to ours for essentially the same reasons as why the Buhrman et al.~result is incomparable: a
different noise model and a different statement about the resulting power of their noisy quantum circuits.

Finally, it is known that it is impossible to transmit quantum information through a $p$-depolarizing channel
for $p>1/3$~\cite{bdefms:cloning}. This seems to suggest that quantum computation over and above classical
computation is impossible with depolarizing noise of strength greater than~$1/3$, but there is no proof that
this is indeed the case.

\section{Preliminaries}\label{sec:prelim}

Let $\calP=\{I,X,Y,Z\}$ be the set of one-qubit Pauli matrices,
$$
I=\left(\begin{array}{cc}1 & 0\\ 0 & 1\end{array}\right), \ X=\left(\begin{array}{cc}0 &
1\\ 1 & 0\end{array}\right), \ Y=\left(\begin{array}{cc}0 & -i\\ i & 0\end{array}\right), \
Z=\left(\begin{array}{cc}1 & 0\\ 0 & -1\end{array}\right).
$$
and let $\calP_*=\{X,Y,Z\}$. We use $\calP^n$ to denote the set of all
tensor products of $n$ one-qubit Pauli matrices. For a Pauli
matrix $S \in \calP^n$ we define its {\em support}, denoted $\supp(S)$, to be the qubits on which $S$
is not identity.
We sometimes use superscripts to indicate the qubits on which certain operators act.
Thus $I^{\cal A}$ denotes the identity operator applied to the qubits in set $\cal A$.

The set of all $2^n \times 2^n$ Hermitian matrices forms a $4^n$-dimensional
real vector space. On this space we consider the Hilbert-Schmidt inner product,
given by $\inpc{A}{B} := \Tr(A^\dag B) = \Tr(AB)$.
Note that for any $S,S'\in\calP^n$, $\Tr(SS')=2^n$ if $S=S'$ and $0$ otherwise, and hence $\calP^n$
is an orthogonal basis of this space. It follows that we can uniquely express any
Hermitian matrix $\delta$ in this basis as
$$
\delta=\frac{1}{2^n}\sum_{S\in\calP^n}\widehat{\delta}(S)S
$$
where $\widehat{\delta}(S):=\Tr(\delta S)$ are the (real) coefficients.

We now state some easy observations which will be used in the proof of our main result. First, by the orthogonality
of $\calP^n$, it follows that for any $\delta$,
$$
\Tr(\delta^2) =  \frac{1}{2^n} \sum_{S\in\calP^n}\widehat{\delta}(S)^2.
$$
This easily leads to the following observation.

\begin{observation}[Unitary preserves sum of squares]\label{obs:unitary}
For any unitary matrix $U$ and any Hermitian matrix $\delta$, if we denote
$\delta' = U \delta U^\dagger$, then
\begin{align*}
\sum_{S \in \calP^n}\widehat{\delta'}(S)^2 =
2^n \Tr(\delta'^2) =
2^n \Tr(U \delta U^\dagger U \delta U^\dagger) =
2^n \Tr(\delta ^2) =
\sum_{S \in \calP^n}\widehat\delta(S)^2.
\end{align*}
\end{observation}

\noindent
This also shows that the operation of conjugating by a unitary matrix, when viewed as a linear
operation on the vector of Pauli coefficients, is an orthogonal transformation.

\begin{observation}[Tracing out qubits]\label{obs:tracing}
Let $\delta$ be some Hermitian matrix on a set of qubits $W$. For $V \subseteq W$, let
$\delta_V=\Tr_{W \setminus V}(\delta)$. Then,
$$\widehat{\delta}(SI^{W \setminus V}) = \Tr( \delta \cdot SI^{W \setminus V})=
\Tr( \delta_V \cdot S) = \widehat{\delta_V}(S).$$
\end{observation}

\begin{observation}[Noise in the Pauli basis]\label{obs:noise}
Applying a $p$-depolarizing noise $\cal E$ to the $j$-th qubit of Hermitian matrix $\delta$
changes the coefficients as follows:
\begin{align*}
\widehat{{\cal E}(\delta)}(S)=\left\{\begin{array}{rl}
\widehat{\delta}(S) & \mbox{ if }S_j=I\\
(1-p)\widehat{\delta}(S) & \mbox{ if }S_j\neq I\\
\end{array}\right.
\end{align*}
\end{observation}
\noindent
In other words, $\cal E$ ``shrinks'' by a factor $1-p$ all coefficients that have support on the $j$-th
coordinate.

\begin{observation}\label{obs:outcome}
Let $\rho$ and $\tau$ be two one-qubit states and let $\delta = \rho-\tau$.
Consider the two probability distributions obtained by performing a measurement in
the computational basis on $\rho$ and $\tau$, respectively. Then the variation distance between
these two distributions is $\frac{1}{2}|\widehat{\delta}(Z)|$.
\end{observation}

\begin{proof}
Since there are only two possible outcomes for the measurements, the variation distance between the two
distributions is exactly the difference in the probabilities of obtaining the outcome $0$,
which is given by
\begin{align*}
|\Tr((\rho - \tau) \cdot\ketbra{0}{0})|
= \left|\Tr\left(\delta\cdot\frac{I+Z}{2}\right) \right|
 = \frac{1}{2}|\Tr(\delta \cdot Z)|= \frac{1}{2}|\widehat{\delta}(Z)|,
\end{align*}
where we have used $\Tr(\delta)=0$.
\end{proof}

Our final observation follows immediately from the convexity of the function $x^2$.
\begin{observation}[Convexity]\label{obs:convexity}
Let $p_i$ be any probability distribution, and $\delta_i$ a set of Hermitian matrices.
Let $\delta=\sum_i p_i \delta_i$. Then
$$
\sum_{S \in \calP^n} \widehat\delta(S)^2 \le  \sum_i p_i \sum_{S \in \calP^n} \widehat\delta_i(S)^2.
$$
\end{observation}

\section{Proof of Theorem \ref{the:main}}\label{secmain}

In this section we prove Theorem \ref{the:main}. The rough idea is the following.
Fix two arbitrary initial states $\rho$ and $\tau$. Our goal is to show that
after applying the noisy circuit, the state of the output qubit is nearly
the same with both starting states. Equivalently, we can define
$\delta = \rho - \tau$ and show that after applying the noisy circuit to
$\delta$, the ``state" of the output qubit is essentially $0$ (notice
that we can view the noisy circuit as a linear operation, and hence
there is no problem in applying it to $\delta$, which is the difference of two density matrices).
In order to show this, we will examine how the coefficients of $\delta$ in the Pauli
basis develop through the circuit. Initially we might have many large coefficients.
Our goal is to show that the coefficients of the output qubit are essentially~0.
This is established by analyzing the balance between two opposing forces:
noise, which shrinks coefficients by a constant factor (as in Observation~\ref{obs:noise}),
and gates, which can increase coefficients. As we saw in Observation~\ref{obs:unitary},
unitary gates preserve the sum of squares of coefficients. They can, however, ``concentrate"
several small coefficients into one large coefficient. One-qubit operations need not preserve
the sum of squares (a good example is the gate that resets a qubit to the $\ket{0}$ state),
but we can still deal with them by using a known characterization of one-qubit gates.
This characterization allows us to bound the amount by which one-qubit gates
can increase the Pauli coefficients, and very roughly speaking shows that
the gate that resets a qubit to $\ket{0}$ is ``as bad as it gets".

Before continuing with the proof, we introduce some terminology.
{}From now on we use the term \emph{qubit} to mean a wire at a specific time, so there are
$(T+1)n$ qubits (although during the proof we will also consider qubits
that are located between a gate and its associated noise).
We say that a set of qubits $V$ is {\em consistent} if
we can meaningfully talk about a ``state of the qubits of $V$" (see Figure \ref{Fig:circuit}).
More formally, we define a consistent set as follows.
The set of all qubits at time $0$ and all its subsets are consistent. If $V$ is some consistent set
of qubits, which contains all input qubits $IN$ of some gate (possibly a one-qubit identity gate),
then also $(V\setminus IN) \cup OUT$ and all
its subsets are consistent, where $OUT$ denotes the gate's output qubits. Note that
here we think of the noise as being part of the gate.
For a consistent set $V$ and a state (or more generally, a Hermitian matrix) $\rho$,
we denote the state of $V$ when the circuit is applied with the initial state $\rho$,
by $\rho_V$. In other words, $\rho_V$ is the state one obtains by applying some initial
part of the circuit to $\rho$, and then tracing out from the resulting state all qubits that are not in $V$

If $v$ is a qubit, we use $\dist(v)$ to denote its distance
from the input, i.e., the level of the gate just preceding it.
The qubits of the starting state have $\dist(v)=0$.
For a nonempty set $V$ of qubits we define $\dist(V)=\min\{\dist(v)\mid
v\in V\}$, and extend it to the empty set by $\dist(\emptyset)=\infty$.
Note that $\dist(V)$ does not increase if we add qubits to $V$.

In the rest of this section we prove the following lemma, showing
that a certain invariant holds for all consistent sets $V$.
\begin{lemma}\label{lem:deltaT}
For all $\eps_1>0$ and $\eps_k > 1-\sqrt{2^{1/k}-1}$ there exists a $\theta <1$ such that
the following holds. Fix any $T$-level circuit in our model, let $\rho$ and $\tau$
be some arbitrary initial states, and let $\delta=\rho-\tau$.
Then for every consistent  $V$,
 \begin{equation}\label{eq:invariant}
 \sum_{S\in\calP^V}\widehat{\delta_V}(S)^2 \leq 2\cdot 2^{|V|}\cdot \theta^{\dist(V)},
 \end{equation}
 or equivalently,
 \begin{align*}
 \Tr(\delta_V^2) \leq 2 \cdot \theta^{\dist(V)}.
 \end{align*}
\end{lemma}

\noindent
In particular, if we consider the consistent set $V$ that contains the designated
output qubit at time $T$, then we get that $\widehat{\delta_V}(Z)^2 \leq 4\theta^T$.
By Observation~\ref{obs:outcome}, this implies Theorem~\ref{the:main}.

\subsection{Proof of Lemma~\ref{lem:deltaT}}\label{ssecproofinvariant}

The proof of the invariant is by induction on the sets $V$. At the base of the induction are all sets
$V$ contained entirely within time $0$. All other sets are handled in the induction step.
In order to justify the inductive proof, we need to provide an ordering on the consistent sets $V$
such that for each $V$, the proof for $V$ uses the inductive hypothesis only on sets $V'$
that appear before $V$ in the ordering.
As will become apparent from the proof, if we denote by ${\rm latest}(V)$ the maximum time at which $V$ contains a qubit,
then each $V'$ for which we use the induction hypothesis has strictly less qubits than $V$ at time ${\rm latest}(V)$.
Therefore, we can order the sets $V$ first in increasing order of ${\rm latest}(V)$ and then in
increasing order of the number of qubits at time ${\rm latest}(V)$.

\subsubsection{Base case}

Here we consider the case that $V$ is fully contained within time $0$.
If $V=\emptyset$ then both sides of the invariant are zero, so from now
on assume $V$ is nonempty. In this case $\dist(V)=0$.
The matrix $\delta_V$ is the difference of two density matrices, say $\delta_V = \rho_V - \tau_V$, and hence
$\Tr(\delta_V^2) = \Tr(\rho_V^2) + \Tr(\tau_V^2) - 2 \Tr(\rho_V \tau_V) \le 2$,
and the invariant is satisfied.

\subsubsection{Induction step}

Let $V''$ be any consistent set containing at least one qubit at time greater than zero.
Our goal in this section is to prove the invariant for $V''$.
Consider any of the qubits of $V$ located at time ${\rm latest}(V)$ and let $G$ be the gate
that has this qubit as one of its output qubits.
We now consider two cases, depending on whether $G$ is a $k$-qubit gate or a one-qubit gate.

\paragraph{\underline{Case 1}: $G$ is a $k$-qubit gate.} Here we consider the case that $G$ is a probabilistic
mixture of $k$-qubit unitaries. First note that by Observation~\ref{obs:convexity} it suffices
to prove the invariant for $k$-qubit unitaries. So assume $G$ is a $k$-qubit unitary acting
on the qubits $\calA=\{A_1,\dots,A_k\}$. Let $\calA'=\{A_1',\dots, A_k'\}$ be the qubits
after the $\epsilon_k$-noise but before the gate $G$ and
$\calA''=\{A_1'',\dots,A_k''\}$ the qubits after $G$ (see Figure~\ref{fig:induction}). By our choice of $G$,
$\calA'' \cap V''\neq \emptyset$. Define $V'=(V'' \setminus \calA'') \cup\calA'$ and $V=(V''
\setminus \calA'') \cup \calA$. Note that $V$ and its subsets are consistent sets with strictly fewer
qubits than $V''$ at time ${\rm latest}(V'')$, and hence we can apply the induction hypothesis to them.

\begin{figure}[h]
\center{\epsfxsize=4in \epsfbox{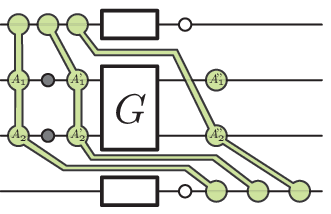}} \caption{An example showing the sets $V$, $V'$,
and $V''$ for a two-qubit gate $G$.}
\label{fig:induction}
\end{figure}

Recall that our goal is to prove the invariant Eq.~(\ref{eq:invariant}) for $V''$. To begin with, using
Observation~\ref{obs:tracing},
\begin{equation}\label{eq:lefthand}
 \sum_{S\in\calP^{V''}}\widehat{\delta_{V''}}(S)^2 \leq \sum_{S\in\calP^{V'' \cup \calA''}}
 \widehat{\delta_{V''\cup \calA''}}(S)^2.
  \end{equation}
Because $G$ (which maps $\delta_{V'}$ to $\delta_{V'' \cup \calA''}$) is unitary, it preserves the
sum of squares of $\widehat{\delta}$-coefficients (see Observation~\ref{obs:unitary}), so the right hand side of
\eqref{eq:lefthand} is equal to
\begin{align*}
\sum_{S\in\calP^{V'}}\widehat{\delta_{V'}}(S)^2=
 \sum_{S\in\calP^{V' \setminus  \calA'}~}\sum_{R\in\calP^{\calA'}~}\widehat{\delta_{V'}}(RS)^2.
\end{align*}
Since the only difference between $\delta_V$ and $\delta_{V'}$ is noise on the qubits $A_1,\dots,A_k$, using
Observation \ref{obs:noise} and denoting $\mu=1-\epsilon_k$, we get that the above is at most
\begin{align*}
& \sum_{S\in\calP^{V \setminus  \calA}}
\sum_{R\in\calP^\calA}
\mu^{2|\supp(R)|}
\widehat{\delta_V}(RS)^2 \\
=& \sum_{S\in\calP^{V \setminus \calA}}
\sum_{a \subseteq \calA}
\mu^{2|a|}(1-\mu^2)^{k-|a|}
\sum_{ R\in\calP^a\otimes I^{\calA \setminus a}} \widehat{\delta_V}(RS)^2,
\end{align*}
where the equality follows by noting that for any fixed $S$ and any $R \in \calP^\calA$, the term
$\widehat{\delta_V}(RS)^2$, which appears with coefficient $\mu^{2|\supp(R)|}$ on the left hand side,
appears with the same coefficient $\sum_{a \supseteq \supp(R)} \mu^{2|a|}(1-\mu^2)^{k-|a|} =
\mu^{2|\supp(R)|}$ on the right hand side.
By rearranging and using Observation \ref{obs:tracing} we get that the above is equal to
\begin{align*}
& \sum_{a \subseteq \calA} \mu^{2|a|}(1-\mu^2)^{k-|a|}\sum_{S\in\calP^{(V \setminus
\calA)\cup a} ~~~} \widehat{\delta_{(V \setminus \calA)\cup a}}(S)^2 \\
\leq  & \sum_{a \subseteq \calA} \mu^{2|a|}(1-\mu^2)^{k-|a|}2\cdot 2^{|(V  \setminus  \calA)\cup a|}
\cdot\theta^{\dist((V  \setminus  \calA)\cup a)}
\end{align*}
where we used the inductive hypothesis.
Note that $\dist((V \setminus \calA)\cup a) \ge \dist(V)$, so the above is
\begin{align}\label{eq:proofend}
\leq  & ~2\cdot 2^{|V \setminus \calA|}\cdot\theta^{\dist(V)}\sum_{a \subseteq
\calA} 2^{|a|}\mu^{2|a|}(1-\mu^2)^{k-|a|} \nonumber \\
= & ~2\cdot 2^{|V \setminus \calA|}\cdot\theta^{\dist(V)}((1-\mu^2)+2\mu^2)^k \nonumber \\
= & ~2\cdot 2^{|V \setminus \calA|}\cdot\theta^{\dist(V)}(1+\mu^2)^k.
\end{align}
Note that $|V \setminus \calA| \leq |V''|-1$ and $\dist(V'')-1 \le \dist(V)$,
so the right hand side is bounded by
$$\leq 2 \cdot 2^{|V''|-1}\cdot\theta^{\dist(V'')-1}(1+\mu^2)^k.$$
Since $\epsilon_k > 1-\sqrt{2^{1/k}-1}$, we have that $(1+\mu^2)^k \le 2 \theta$
if $\theta$ is close enough to $1$, so we can finally bound the last expression by
$$\leq 2 \cdot 2^{|V''|}\cdot\theta^{\dist(V'')}$$ which proves the invariant for $V''$.

\paragraph{\underline{Case 2}: $G$ is a one-qubit gate.}
Before proving the invariant, we need to prove the following property of
completely-positive trace-preserving (CPTP) maps on one qubit.
\begin{lemma}\label{lem:onequbit}
For any CPTP map $G$ on one qubit there exists a $\beta\in[0,1]$ such that
the following holds. For any Hermitian matrix $\delta$, if we let $\delta'$ denote
the result of applying $G$ to $\delta$, then we have
$$ \widehat{\delta'} (X)^2 + \widehat{\delta'} (Y)^2 + \widehat{\delta'} (Z)^2 \le
   (1-\beta) \cdot \widehat{\delta}(I)^2 + \beta \cdot ( \widehat{\delta}(X)^2 + \widehat{\delta}(Y)^2 + \widehat{\delta}(Z)^2).$$
\end{lemma}
\begin{proof}
The proof is based on the characterization of trace-preserving completely-positive maps on one qubit
due to Ruskai, Szarek, and Werner~\cite[Sections~1.2 and~1.3]{rsw:cpqubitmaps}.
This characterization implies that any one-qubit gate $G$ can be written as a convex combination of gates of the form
$U_1\circ J\circ U_2$. Here $U_1$ and $U_2$ are one-qubit unitaries (acting on the density matrix by conjugation),
and $J$ is a one-qubit map that in the Pauli basis has the form
$$
J=\left(\begin{array}{cccc}
1 & 0 & 0 & 0\\
0 & \lambda_1 & 0 & 0\\
0 & 0 & \lambda_2 & 0\\
t & 0 & 0 & \lambda_1\lambda_2
\end{array}\right)
$$
for some $\lambda_1,\lambda_2\in[-1,1]$ and $t=\pm \sqrt{(1-\lambda_1^2)(1-\lambda_2^2)}$.

First observe that by the convexity of the square function, it suffices to prove the lemma
for $G$ of the form $U_1\circ J\circ U_2$ (with the resulting $\beta$ being the appropriate
average of the individual $\beta$'s). Next note that since $U_1$ and $U_2$ are unitary,
they act on the vector of coefficients $(\widehat{\delta}(X), \widehat{\delta}(Y), \widehat{\delta}(Z))$
as an orthogonal transformation, and hence leave the sum of squares invariant. This shows
that it suffices to prove the lemma for a map $J$ as above. For this map,
\begin{align*}
\widehat{\delta'} (X)^2 + \widehat{\delta'} (Y)^2 + \widehat{\delta'} (Z)^2 =
\lambda_1^2 \widehat{\delta} (X)^2 + \lambda_2^2 \widehat{\delta}(Y)^2 + (t \widehat{\delta}(I) + \lambda_1\lambda_2 \widehat{\delta}(Z))^2.
\end{align*}
Assume without loss of generality
that $\lambda_1^2\geq\lambda_2^2$. Applying Cauchy-Schwarz to the two 2-dimensional vectors $(\pm \sqrt{1-\lambda_1^2}a,\lambda_1 b)$ and
$(\sqrt{1-\lambda_2^2},\lambda_2)$, we get that for any $a,b \in \R$, $(t
a+\lambda_1\lambda_2 b)^2\leq (1-\lambda_1^2)a^2 + \lambda_1^2 b^2$. Hence the above
expression is upper bounded by
$$
\lambda_1^2 \widehat{\delta} (X)^2 + \lambda_1^2 \widehat{\delta}(Y)^2 + (1-\lambda_1^2)\widehat{\delta}(I)^2 + \lambda_1^2 \widehat{\delta}(Z)^2
$$
and we complete the proof by choosing $\beta = \lambda_1^2$.
\end{proof}

Let $A$ be the qubit $G$ is acting on, and recall that our goal is to prove the invariant
for the set $V''$. Denote by $A'$ the qubit
of $G$ after the gate but before the $\eps_1$ noise, and by $A''$ the qubit after the noise. As before,
by our choice of $G$, we have $A'' \in V''$. Let $\calA = \{A\}$, $\calA' = \{A'\}$, $\calA'' = \{A''\}$.
Define $V'=(V'' \setminus \calA'') \cup
\calA'$ and $V=(V'' \setminus \calA'') \cup \calA$ and notice that $|V|=|V'|=|V''|$.
By using Lemma~\ref{lem:onequbit}, we obtain a $\beta\in[0,1]$ such that
\begin{align*}
\sum_{S\in\calP^{V''}}\widehat{\delta_{V''}}(S)^2  &\le
\sum_{S\in\calP^{V' \setminus
\calA'}}\Bigg(\widehat{\delta_{V'}}(I S)^2+
(1-\epsilon_1)^2\sum_{R\in\calP_*^{\calA'}}\widehat{\delta_{V'}}(RS)^2\Bigg)\\
& \le \sum_{S\in\calP^{V \setminus \calA}}\Bigg((1+(1-\epsilon_1)^2(1-2\beta))\widehat{\delta_V}(I
S)^2+
(1-\epsilon_1)^2 \beta \sum_{R\in\calP^\calA}\widehat{\delta_V}(RS)^2\Bigg).
\end{align*}
By applying the induction hypothesis to both $V \setminus \calA$ and $V$, we can
upper bound the above by
\begin{align*}
& (1+(1-\epsilon_1)^2(1-2\beta))\cdot 2\cdot 2^{|V|-1}\cdot\theta^{\dist(V \setminus \calA)}
+(1-\epsilon_1)^2 \beta \cdot 2\cdot 2^{|V|}\cdot\theta^{\dist(V)}\\
&\leq \frac{1+(1-\epsilon_1)^2}{2\theta}\cdot 2\cdot 2^{|V''|}\cdot\theta^{\dist(V'')}
\end{align*}
where we used that $|V|=|V''|$, and $\dist(V'') - 1\leq \dist(V) \leq \dist(V \setminus \calA)$.
Hence the invariant remains valid if we choose $\theta<1$ such that $1+(1-\epsilon_1)^2\leq 2\theta$.

\section{Arbitrary one-qubit gates and CNOT gates}\label{seccnot}

In this section we consider the case where CNOT is the only allowed gate acting on more than one qubit. We still allow
arbitrary one-qubit gates. The proof follows along the lines of that of Theorem~\ref{the:main} with one small modification.
As before, we will prove that for all $\eps_1>0$ and $\eps_2>1-1/\sqrt{2} \approx
0.293$ the invariant, Eq.~(\ref{eq:invariant}), holds. The proof for the case that $G$ is a one-qubit gate
holds without change. We will give the modified proof for the case that $G$ is a CNOT gate. The idea for the
improved bound is to make use of the fact that the CNOT gate merely permutes the 16 elements of
$\cal P\otimes\cal P$, and does not map elements from $I\otimes \calP_*$ to
$\calP_*\otimes I$ or vice versa (as illustrated in Figure~\ref{fig:cnot}).
As a result we need to apply the induction hypothesis on one less term, which in turn
improves the bound.

\begin{figure}[h]
\center{\epsfxsize=1.6in \epsfbox{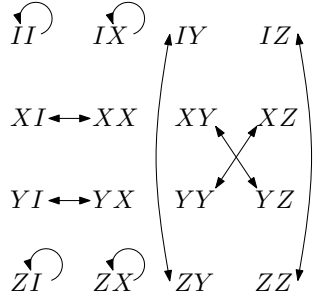}}
\caption{The action of CNOT on $\calP\otimes\calP$ under conjugation with
the control wire corresponding to the first qubit.}
\label{fig:cnot}
\end{figure}

Assume the CNOT acts on qubits $\calA=\{A,B\}$, with $\calA'=\{A',B'\}$ and $\calA''=\{A'',B''\}$ as before, where again $\calA'' \cap V'' \neq \emptyset$.
If both $A''$ and $B''$ are contained in $V''$ then the proof of the general case (cf. Eq.~(\ref{eq:proofend}))
already gives a bound of
$$
2 \cdot 2^{|V \setminus \calA|} \cdot \theta^{\dist(V)} (1+\mu^2)^2\leq 2\cdot 2^{|V''|-2}\cdot
\theta^{\dist(V'')-1}(1+\mu^2)^2 \leq 2\cdot 2^{|V''|}\cdot \theta^{\dist(V'')}$$
where the last inequality holds for all $\mu <1$. Hence it suffices to consider the case
that exactly one of $A''$ and $B''$ is in $V''$. Assume without loss of generality
that $A'' \in V''$ and $B'' \notin V''$. As before, our goal is to upper bound
$$
\sum_{S\in\calP^{V''}}\widehat{\delta_{V''}}(S)^2=\sum_{S\in\calP^{V''}}\widehat{\delta_{V'' \cup B''}}(SI^{B''})^2,
$$
where the equality follows from Observation~(\ref{obs:tracing}).
Because of the property of CNOT mentioned above, we can now upper bound this by
\begin{align*}
\sum_{S\in\calP^{V' \setminus \calA'}} \Big( \widehat{\delta_{V'}}(I^{A'} I^{B'} S)^2
+\sum_{R\in\calP_*^{A'}}\widehat{\delta_{V'}}(R I^{B'} S)^2 +\sum_{R\in\calP_*^{A'}\otimes\calP_*^{B'}}\widehat{\delta_{V'}}(R S)^2
\Big).
\end{align*}
This is the crucial change compared to the case of general two-qubit gates (the latter case also includes
a term of the form $\sum_{R\in\calP_*^{B'}}\widehat{\delta_{V'}}(I^{A'}R S)^2$).
The rest of the proof is similar to the earlier proof. Using the induction hypothesis we can upper bound
the above by
\begin{align*}
&\sum_{S\in\calP^{V \setminus \calA}} \Big( \widehat{\delta_V}(I^A I^B S)^2 +\mu^2
\sum_{R\in\calP_*^A}\widehat{\delta_V}(R I^B S)^2 +\mu^4 \sum_{R\in\calP_*^A\otimes\calP_*^B}\widehat{\delta_V}(R S)^2
\Big)\\
& \leq ~(1-\mu^2)\sum_{S\in\calP^{V \setminus \calA}} \widehat{\delta_{V\setminus \calA}}(S)^2 +
(\mu^2-\mu^4)\sum_{S\in\calP^{V \setminus \{B\}}}\widehat{\delta_{V \setminus \{B\}}}(S)^2 + \mu^4 \sum_{S\in
\calP^{V}}\widehat{\delta_V}(S)^2
\big)\\
&\leq ~(1-\mu^2)2\cdot 2^{|V\setminus \calA|}\cdot\theta^{\dist(V \setminus \calA)} +
(\mu^2-\mu^4)2\cdot 2^{|V\setminus \{B\}|}\cdot\theta^{\dist(V\setminus \{B\})}
+ \mu^4~2\cdot 2^{|V|}\cdot\theta^{\dist(V)}\\
& \leq ~2\cdot 2^{|V''|}\cdot\theta^{\dist(V)}\Big(\frac{1+\mu^2}{2}+\mu^4\Big)\\
& \leq ~2\cdot 2^{|V''|}\cdot\theta^{\dist(V'')}\Big({\frac{1+\mu^2}{2}+\mu^4}\Big)\frac{1}{\theta}.
\end{align*}
Hence the invariant remains valid as long as $\frac{1+\mu^2}{2}+\mu^4\leq\theta<1$.
This can be satisfied as long as $\mu<1/\sqrt{2}$, equivalently $\eps_2>1-1/\sqrt{2}\approx 0.293$.

\subsubsection*{Acknowledgment}
We thank Mary Beth Ruskai for a pointer to~\cite{rsw:cpqubitmaps} and for sharing her insights
on the characterization of one-qubit operations.  We thank Peter Shor for
a discussion on entanglement breaking channels which is related to the discussion of~\cite{bdefms:cloning} at
the end of Section \ref{secintro}.


\end{document}